\documentclass[submission]{eptcs}
\providecommand{\event}{QPL 2019} 
\usepackage{underscore}           

\usepackage{amsfonts}
\usepackage[fleqn]{amsmath}
\usepackage{amssymb}
\usepackage{array}
\usepackage{hyperref}
\usepackage{amsthm}
\usepackage{mathtools}
\usepackage[svgnames]{xcolor}
\usepackage{tikz}
\usepackage{stmaryrd}
\usepackage{longtable}
\usepackage{setspace}
\usepackage[utf8]{inputenc}
\usepackage[section]{placeins}
\usepackage{booktabs}
\usepackage{csquotes}
\usepackage{comment}
\usepackage{thmtools}
\usepackage{thm-restate}
\usepackage{cleveref}

\DeclareUnicodeCharacter{27E9}{\rangle}
\DeclareUnicodeCharacter{3C0}{\pi}

\def\Mat{\text{Mat}}

\def\Hilbert{\mathbb{H}}
\def\bbC{\mathbb{C}}
\def\bbN{\mathbb{N}}
\def\bbL{\mathbb{L}}
\def\bbD{\mathbb{D}}
\def\bbE{\mathbb{E}}
\def\ZW{ZW$_\bbC$\,}
\def\id{\text{id}}
\def\issound{\text{ is sound}}

\newtheorem{theorem}{Theorem}
\newtheorem{definition}[theorem]{Definition}

\newtheorem{proposition}[theorem]{Proposition}
\newtheorem{corollary}[theorem]{Corollary}
\newtheorem{example}[theorem]{Example}
\newtheorem{remark}[theorem]{Remark}

\makeatletter
\AtBeginDocument{%
  \expandafter\renewcommand\expandafter\subsection\expandafter{%
    \expandafter\@fb@secFB\subsection
  }%
}
\makeatother

\graphicspath{ {./images/} }

\newcommand\ket[1]{|\, #1 \, \rangle}

\newcommand{\node}[2]{
	\begin{tikzpicture}[scale=0.3,
    every node/.style={scale=0.75}]
	\begin{pgfonlayer}{nodelayer}
		\node [style=#1] (1)   at (0, 0)   {$#2$};
		\node [style=none] (2)   at (0, 1)   {};
		\node [style=none] (3)   at (0, -1)   {};
	\end{pgfonlayer}
	\begin{pgfonlayer}{edgelayer}
		\draw [-] (1.center) to (2.center);
		\draw [-] (1.center) to (3.center);
	\end{pgfonlayer}
	\end{tikzpicture}
}

\newcommand{\spider}[2]{
	\begin{tikzpicture}[scale=0.3,
    every node/.style={scale=0.75}]
	\begin{pgfonlayer}{nodelayer}
		\node [style=#1] (1)   at (0, 0)   {$#2$};
		\node [style=none] (2)   at (-1, 2)   {};
		\node [style=none] ()   at (0, 2)   {$\dots$};
		\node [style=none] (3)   at (1, 2)   {};
		\node [style=none] (4)   at (-1, -2)   {};
		\node [style=none] ()   at (0, -2)   {$\dots$};
		\node [style=none] (5)   at (1, -2)   {};
	\end{pgfonlayer}
	\begin{pgfonlayer}{edgelayer}
		\draw [-] (1.center) to (2.center);
		\draw [-] (1.center) to (3.center);
		\draw [-] (1.center) to (4.center);
		\draw [-] (1.center) to (5.center);
	\end{pgfonlayer}
	\end{tikzpicture}
}

\newcommand{\wspider}{
	\begin{tikzpicture}[scale=0.3,
    every node/.style={scale=0.75}]
	\begin{pgfonlayer}{nodelayer}
		\node [style=black] (1)   at (0, 0)   {};
		\node [style=none] (2)   at (-1, 2)   {};
		\node [style=none] ()   at (0, 2)   {$\dots$};
		\node [style=none] (3)   at (1, 2)   {};
	\end{pgfonlayer}
	\begin{pgfonlayer}{edgelayer}
		\draw [-] (1.center) to (2.center);
		\draw [-] (1.center) to (3.center);
	\end{pgfonlayer}
	\end{tikzpicture}
}

\newcommand{\zxspider}[2]{
\begin{tikzpicture}
	\begin{pgfonlayer}{nodelayer}
		\node [style=gn] (0) at (-1, 0.25) {$#1$};
		\node [style=gn] (1) at (-1, -0.25) {$#2$};
		\node [style=gn] (2) at (1, 0) {$#1 + #2$};
		\node [style=none] (3) at (0, 0) {=};
		\node [style=none] (4) at (-1.5, 0.75) {};
		\node [style=none] (5) at (-0.5, 0.75) {};
		\node [style=none] (6) at (-1.5, -0.75) {};
		\node [style=none] (7) at (-1, -0.75) {};
		\node [style=none] (8) at (-0.5, -0.75) {};
		\node [style=none] (9) at (0.5, -0.75) {};
		\node [style=none] (10) at (1, -0.75) {};
		\node [style=none] (11) at (1.5, -0.75) {};
		\node [style=none] (12) at (0.5, 0.75) {};
		\node [style=none] (13) at (1.5, 0.75) {};
	\end{pgfonlayer}
	\begin{pgfonlayer}{edgelayer}
		\draw (4.center) to (0.center);
		\draw (5.center) to (0.center);
		\draw (0.center) to (1.center);
		\draw (1.center) to (6.center);
		\draw (1.center) to (7.center);
		\draw (1.center) to (8.center);
		\draw (12.center) to (2.center);
		\draw (2.center) to (13.center);
		\draw (2.center) to (9.center);
		\draw (2.center) to (10.center);
		\draw (2.center) to (11.center);
	\end{pgfonlayer}
\end{tikzpicture}
}
	
\def\tripleblobmaths{
\begin{tikzpicture}[transform canvas={scale=0.5}]
		\node [style=none] (1)   at (0, 0.5)   {};
		\node [style=none] (2)   at (0, -0.5)   {};
		\draw [-] (1) to (2);
		\draw [bend left=35] (1.center) to (2.center);
		\draw [bend right=35] (1.center) to (2.center);
		\node [style=gn] (g)   at (0, 0.5)   {};
		\node [style=rn] (r)   at (0, -0.5)   {};
	\end{tikzpicture}
}

\def\oneoverrt2{
\tripleblobmaths
\quad
}

\newcolumntype{x}[1]{>{\centering\arraybackslash\hspace{0pt}}p{#1}}

\newcommand{\set}[1]{\{#1\}}
\newcommand{\at}[1]{\Big |_{#1}}
\newcommand{\family}[1]{\left\{ \quad #1 \quad \right\}}
\def\alphadelta{\alpha_1, \dots, \alpha_n, \delta_1, \dots, \delta_m}
\newcommand{\inv}{^{-1}}
\newcommand{\abs}[1]{|#1|}
\newcommand{\comp}{\circ}
\newcommand{\interpret}[1]{\left\llbracket \; #1 \; \right\rrbracket}
\def\tensor{\ensuremath\otimes}

\def\iso{\cong}

\usetikzlibrary{decorations.pathmorphing}
\usetikzlibrary{decorations.markings}
\usetikzlibrary{decorations.pathreplacing}
\usetikzlibrary{arrows}
\usetikzlibrary{shapes.geometric}

\pgfdeclarelayer{edgelayer}
\pgfdeclarelayer{nodelayer}
\pgfsetlayers{edgelayer,nodelayer,main}

\definecolor{zx_red}{RGB}{232, 165, 165}
\definecolor{zx_green}{RGB}{216, 248, 216}

\tikzstyle{none}=[inner sep=0pt]
\tikzstyle{simple}=[inner sep=0pt]
\tikzstyle{gn}=[rectangle,rounded corners=0.4em,fill=zx_green,draw=Black,line width=0.4 pt,inner sep=2pt,minimum width=1em,minimum height=1em]
\tikzstyle{rn}=[rectangle,rounded corners=0.4em,fill=zx_red,draw=Black,line width=0.4 pt,inner sep=2pt,minimum width=1em,minimum height=1em]
\tikzstyle{H}=[rectangle,fill=Yellow!30!White,draw=Black]
\tikzstyle{ZH}=[rectangle,fill=White,draw=Black,inner sep=3pt]
\tikzstyle{ZHGrey}=[circle,fill=Gray,draw=Black,inner sep=3pt]
\tikzstyle{white}=[rectangle,rounded corners=0.4em,fill=White,draw=Black,line width=0.4 pt,inner sep=1.5pt,minimum width=1em,minimum height=1em]
\tikzstyle{black}=[circle,fill=black,draw=Black]
\tikzstyle{crossing}=[circle,draw=Black,inner sep=3pt]
\tikzstyle{grey}=[circle,fill=gray,draw=Black]
\tikzstyle{block}=[rectangle,fill=Orange!60!White,draw=Black]
\tikzstyle{plain}=[rectangle,inner sep=3pt]
\tikzstyle{box}=[draw, rectangle,inner sep=3pt]
\tikzstyle{dtriangle}=[fill=yellow,draw=black,shape=isosceles triangle,shape border rotate=-90,isosceles triangle stretches=true,inner sep=1pt,minimum width=1em,minimum height=1em]
\tikzstyle{vtriang}=[fill=yellow,draw=black,shape=isosceles triangle,shape border rotate=180,isosceles triangle stretches=true,inner sep=1pt,minimum width=1em,minimum height=1em]
\tikzstyle{triangle}=[fill=yellow,draw=black,shape=isosceles triangle,shape border rotate=90,isosceles triangle stretches=true,inner sep=1pt,minimum width=1em,minimum height=1em]
\tikzstyle{empty}=[rectangle,dashed,draw=Black,minimum width=1.5em,minimum height=1.5em]
\tikzstyle{bbox}=[rectangle,dashed,draw=Black,minimum width=1.5em,minimum height=1.5em]
\tikzstyle{net}=[rectangle,minimum width=1.5em,minimum height=1.5em]
\tikzstyle{arrow}=[->,draw=black]
\tikzstyle{light-arrow}=[->,draw=gray]

\tikzstyle{every picture}=[baseline=-0.25em,scale=0.75, every node/.style={transform shape}]
\tikzset{global scale/.style={
    scale=0.5,
    every node/.style={scale=0.5}
  }
}

\usetikzlibrary{arrows, decorations.markings, shapes.geometric, decorations.pathmorphing, decorations.pathreplacing, intersections, patterns,calc, backgrounds}

\pgfdeclarelayer{bg}
\pgfdeclarelayer{mid}
\pgfdeclarelayer{fg}
\pgfsetlayers{bg,mid,fg,main,edgelayer,nodelayer}

\tikzset{
	0c/.style={circle, draw, fill, inner sep=1.5pt},
	1c/.style={->, thick, shorten <=2pt, shorten >=2pt},
	1cboth/.style={<->, thick, shorten <=2pt, shorten >=2pt},
	1clong/.style={->, thick},
	1cthin/.style={->, shorten <=4pt, shorten >=4pt},
	1cdot/.style={->, dashed, thick, shorten <=2pt, shorten >=2pt},
	1cinc/.style={right hook->, thick, shorten <=2pt, shorten >=2pt},
	1cincl/.style={left hook->, thick, shorten <=2pt, shorten >=2pt},
	follow/.style={->, >=stealth, very thick, shorten <=3pt, shorten >=3pt, color=magenta},
	2c/.style={double, thick, shorten <=6pt, shorten >=8pt, decoration={markings,mark=at position -6pt with {\arrow[scale=1.75]{>}}}, preaction={decorate}},
	2cdot/.style={double, dashed, thick, shorten <=10pt, shorten >=10pt, decoration={markings,mark=at position -8pt with {\arrow[scale=1.75]{>}}}, preaction={decorate}},
	3c1/.style={thick, double, double distance=3pt, shorten <=9pt, shorten >=11pt},
    	3c2/.style={thick, shorten <=9pt, shorten >=10pt},
	3c3/.style={shorten <=9pt, shorten >=10pt, decoration={markings,mark=at position -8pt with {\arrow[scale=3]{>}}},preaction={decorate}},
	4c1/.style={thick, double, double distance=4pt, shorten <=1pt, shorten >=2.75pt},
	4c2/.style={thick, double, double distance=1pt, shorten <=1pt, shorten >=1.25pt, decoration={markings,mark=at position -.05pt with {\arrow[scale=3,ultra thin]{>}}},preaction={decorate}},
	edge/.style={line width=.8pt, color=black},
	edgedot/.style={densely dotted, line width=.8pt, color=black},
	edgethdot/.style={densely dotted, line width=.4pt, color=gray},
	edgeth/.style={line width=.4pt, color=gray!60},
	edgethin/.style={line width=.8pt, color=gray!60},
	edgedotdark/.style={densely dotted, line width=.8pt, color=gray!80},
	dot/.style={circle, draw=black, line width=.8pt, fill=white, inner sep=1.7pt},
	dotth/.style={circle, draw=gray!60, fill=gray!60, inner sep=1.5pt},
	dotwh/.style={circle, draw=gray!60, line width=.4pt, fill=white, inner sep=1.7pt},
	dotwhite/.style={circle, draw=black, line width=.8pt, fill=white, inner sep=1.8pt},
	dotdark/.style={circle, draw, fill=black, inner sep=1.5pt},
	dotgrey/.style={circle, draw=black, line width=.8pt, fill=gray!60, inner sep=1.8pt},
	trian/.style={regular polygon,regular polygon sides=3,shape border rotate=0,fill=white, line width=.8pt, draw=black, inner sep=1.8pt},
	trianh/.style={regular polygon,regular polygon sides=3,shape border rotate=0,fill=white, draw=gray!60, line width=.4pt, inner sep=1.8pt},
	trib/.style={regular polygon,regular polygon sides=3,shape border rotate=0,fill=black, draw, inner sep=1.5pt},
	tribh/.style={regular polygon,regular polygon sides=3,shape border rotate=0,fill=gray!60, draw=gray!60, inner sep=1.5pt},
	tribco/.style={regular polygon,regular polygon sides=3,shape border rotate=180,fill=black, draw, inner sep=1.5pt},
	cover/.style={circle, draw=gray!10, fill=gray!10, inner sep=3.5pt},
	coverc/.style={circle, draw=gray!80, line width=.4pt, inner sep=3pt},
	coverch/.style={circle, draw=gray!30, line width=.4pt, inner sep=3pt},
	coverb/.style={circle, draw=gray!80, line width=.4pt, fill=gray!10, inner sep=3.5pt},
}

\newcommand\bb[1]{\mathbb{#1}}

\newcommand\maxdegree[3]{\text{deg}^{#1}_{#2} \;#3}

\title{Finite Verification of Infinite Families of Diagram Equations}
\author{
Hector Miller-Bakewell
\institute{Department of Computer Science, University of Oxford, Parks Road, Oxford, OX1 3QD}
\email{hector.miller-bakewell@cs.ox.ac.uk}
    } 
\def\titlerunning{Finite Verification}
\def\authorrunning{H. Miller-Bakewell}

\def\bboxes{\mbox{!-boxes}}
\begin{document}
\maketitle

\begin{abstract}
The ZX, ZW and ZH calculi are all graphical calculi for reasoning about
pure state qubit quantum mechanics.
All of these languages use certain diagrammatic decorations,
called !-boxes and phase variables, to indicate
not just one diagram but an infinite family of diagrams.
These decorations are powerful enough to allow complete rulesets
for these calculi to be expressed in around fifteen rules.
Historically rules involving \bboxes\ have not been verifiable by computer.
We present the first algorithm for reducing infinite families of equations
involving \bboxes\ into finite verifying subsets.
The only requirement for this method is a mild property on the connectivity of the \bboxes.
Previous results had focussed on finite case analysis of phase variables in ZX,
a result we also extend for ZW and ZH, as well as providing a general framework for further languages.
The results presented here allow proof assistants to reduce
infinite families of problems (involving combinations of phase variables and !-boxes) down to undecorated, case-by-case verification,
in a way not previously possible.
In particular we note the removal of the need to reason directly with !-boxes
in verification tasks as something entirely new.
This forms part of larger work in automated verification
of quantum circuitry, conjecture synthesis, and diagrammatic languages in general.
The methods described here extend to any diagrammatic languages that meet certain simple conditions.
\end{abstract}

\section{Introduction}

The investigation of qubit graphical calculi began with Coeke and Duncan in \cite{Coecke08},
where they created the ZX calculus:
A sound, universal calculus that was shown to be complete in \cite{UniversalComplete}.
Since then ZX has been used for such things as reasoning about quantum error correction
via lattice surgery
(\cite{2017arXiv170408670D} and \cite{2018arXiv181201238G},)
through to being the basis for taught courses in quantum computing \cite{PQP}.
The ZW$_\bbC$ calculus, invented by Hadzihasanovic \cite{ZW},
presents a different point of view;
rather than focus on the Z and X rotations of the Bloch sphere as ZX does,
it focusses on the GHZ and W entanglement states.
The ZH calculus (Kissinger and Backens, \cite{ZH})
has another viewpoint again; that of extending the notion of Hadamard and CCZ gates.

Each of these calculi has strength in different areas,
and all of them are complete and universal for pure state qubit quantum mechanics.
Our interest is that, with one exception, the rules of these calculi are expressed in a finite manner;
that is to say a finite collection of parameterised families of equations.
This parameterisation is an expression of two types of regular structure:
\begin{itemize}
\item Phase variables: e.g. ``This $X$ can be any complex number''
\item !-boxes: e.g. ``This part of the diagram can be repeated 0 or more times''
\end{itemize}
The one exception to this being the (EU) rule of the ZX calculus (see \cite{VilmartZX})
which use what we call \emph{side conditions},
a case we shall not be considering in this paper.

The aim of this paper is to show when these parameterisations
admit themselves to finite verification:
A process by which a computer can verify an entire infinite family
by checking a finite number of cases.
For phase variables we exploit the
properties of the polynomial functions they represent (theorem \ref{thmparameter}),
and for !-boxes we exploit the finite dimensionality of the
space in which their repeated structure resides (theorem \ref{thmbbox2}).
We show how these two types of parameterisation interact in theorem \ref{thmboth}
before finally giving a constructive method for finding
a verifying subset in theorem \ref{thmfinite}.

In essence finite verification allows us to start with an infinite family of equations
(indicated by a pair of diagrams decorated with phase variables and !-boxes,)
and show that this entire family of equations are sound by only checking
a finite number of individual equations, none of which are decorated.

\section{Parameterised families of diagrams}

In this chapter we define the language with which we discuss decorations, infinite families and so on.
We then give a few examples of the way we use this language.

\begin{definition}  \label{defDiagramFromProp}
A \textbf{diagram} is a morphism in a PROP (see e.g. \cite[p 97]{maclane1965});
it has $n \in \bbN$ inputs (at the top) and $m \in \bbN$ outputs (at the bottom.)
Morphisms in the PROP are constructed from a set of generators using horizontal composition (given the symbol $\tensor$)
and vertical composition (given the symbol $\comp$.)
\end{definition}

\begin{definition} \label{defMatrixInterpretation}
A \textbf{complex matrix interpretation} for a PROP of diagrams $\bbL$, written $\interpret{\cdot}$,
is a monoidal functor from $\bbL$
to $\Mat_\bbC$. I.e. a functor that preserves the $\tensor$ and $\comp$ products of diagrams:
\begin{align}
\interpret{\cdot} & : \bbL \to \Mat_\bbC \\
\interpret{\bbD_1 \tensor \bbD_2} &= \interpret{\bbD_1} \tensor \interpret{\bbD_2} \\
\interpret{\bbD_1 \comp \bbD_2} &= \interpret{\bbD_1} \comp \interpret{\bbD_2}
\end{align}

(We are using $\tensor$ to represent the Kronecker product of matrices,
and $\comp$ to represent standard matrix multiplication.)
\end{definition}

We include a list of the generators of ZX, ZH, and ZW$_\bbC$, as well as their interpretations into $\Mat_{\bbC}$,
in appendix \ref{secgenerators}.

\begin{definition}  \label{defSimpleSound}
An equation of diagrams $\bbD_1 = \bbD_2$ is \textbf{sound} if $\interpret{\bbD_1} = \interpret{\bbD_2}$
\end{definition}

\begin{definition}
A \textbf{phase algebra}
describes additional structure used in presenting infinite families of rules.
Diagram generators are decorated with terms from the phase algebra.
Elements of this algebra are called \textbf{phases} (introduced in \cite{Coecke08}),
and a \textbf{phase variable} is a formal variable adjoined to the phase algebra.
For universal ZX the phase group is $([0, 2\pi), +_{2\pi})$,
and for ZH and ZW$_\bbC$ the phase ring is the complex numbers.
\end{definition}

\begin{example} \label{exaPhaseWithPolynomial}
A ZH diagram that has no phase variables has every phase an element of $\bbC$.
A ZH diagram that has a single phase variable $\alpha$ has every phase an element of $\bbC[\alpha]$.
Unused phase variables (and later, unused !-boxes) are ignored.
\end{example}

\begin{definition}
A \textbf{!-box} (introduced in \cite{Kissinger2012PatternGR}, discussed in more detail in \cite{MerryThesis})
is a special vertex $\delta$ added to the diagram.
A vertex $v$ is described as ``in'' a particular !-box if there is a directed edge from $\delta$ to $v$.
Instantiating the !-box $\delta$ with $n$ copies is performed by:
\begin{itemize}
	\item Making $n$ copies of every vertex in $\delta$, maintaining any connections to other vertices.
	\item Deleting the vertex $\delta$
\end{itemize}
\end{definition}

\begin{definition}  \label{defSimpleDiagram}
A diagram that does not contain any phase variables or !-boxes is called \textbf{simple},
or otherwise called \textbf{decorated}.
\end{definition}

\begin{definition} We will use the term \textbf{instantiation} to refer to both the evaluation of a phase variable
and the copying procedure of !-boxes.
We use the notation $\alpha | \alpha = a$ to indicate that $\alpha$ has been instantiated at the value $a$.
\end{definition}

An explicit instantiation is given in Example \ref{exaExplicitInstantiation}.

\begin{definition}
We write
\begin{align}
\family{\bbD_1 = \bbD_2}_{\alpha_1, \dots, \alpha_n, \delta_1, \dots, \delta_m}
\end{align}
for the family of simple equations between diagrams
parameterised by the $\alpha_j$ and $\delta_k$.
\end{definition}

\begin{definition} The parameterised equation $\bbE$ (between diagrams $\bbD_1$ and $\bbD_2$) is \textbf{sound}
if for any choice of parameter values the resulting equation between simple diagrams is sound.
(See Definition~\ref{defSimpleSound})
\end{definition}

We show restrictions of the parameter values 
by expressing the instantiated value as belonging to some subset.
For example if the family $\bbE$ is sound for $\alpha_2 = a_2$ whenever $a_2 \in A_2$ we will just write:
\begin{align}
\family{\bbE}_{\alpha_1, \dots, \alpha_n, \delta_1,\dots,\delta_m | \alpha_2 \in A_2} \issound
\end{align}

\begin{example} \label{exaExplicitInstantiation}
The (slightly simplified) spider law in universal ZX is parameterised over
\begin{itemize}
\item $\delta_1 \in \bb{N} $ inputs and $\delta_2 \in \bb{N} $ outputs 
\item $\alpha_1, \alpha_2 \in [0, 2\pi)$
\end{itemize}
And so we write the parameterised family of equations as:

\begin{align}
\family{
	
\InputIfFileExists{param-s1-l}{}{\input{./tikz/param-s1-l.tikz}}
 = \; 
\InputIfFileExists{param-s1-r}{}{\input{./tikz/param-s1-r.tikz}}

}_{\alpha_1, \alpha_2, \delta_1, \delta_2}
\end{align}

We now instantiate some of the parameters of the spider law,
resulting in what is still an infinite, parameterised family:

\begin{align}
&\family{
	
\InputIfFileExists{param-s1-l}{}{\input{./tikz/param-s1-l.tikz}}
 = \; 
\InputIfFileExists{param-s1-r}{}{\input{./tikz/param-s1-r.tikz}}

}_{\alpha_1, \alpha_2, \delta_1, \delta_2 | \alpha_1 = \pi, \delta_1 = 2} 
= &\family{
	
\InputIfFileExists{param-s1-l-pi}{}{\input{./tikz/param-s1-l-pi.tikz}}
 = \; 
\InputIfFileExists{param-s1-r-pi}{}{\input{./tikz/param-s1-r-pi.tikz}}

	}_{\alpha_2, \delta_1, \delta_2 | \delta_1 = 2} \\
= &\family{
	
\InputIfFileExists{param-s1-l-pi-2}{}{\input{./tikz/param-s1-l-pi-2.tikz}}
 = \; 
\InputIfFileExists{param-s1-r-pi-2}{}{\input{./tikz/param-s1-r-pi-2.tikz}}

	}_{\alpha_2, \delta_2}
\end{align}

\end{example}

\section{Verifying phase parameters}
\label{secparameter}

Our first result will concern diagrams that contain a finite number of phase variables and no !-boxes.
It relies on a certain property of polynomials: That if you know the value of the polynomial $P(Y)$
for sufficiently many values of $Y$ then you can determine all the coefficients of $P$.
For example the polynomial $P(Y_1, Y_2) = a + b Y_1 + c Y_2 + d Y_1 Y_2$ can have all of its coefficients
determined by knowing the values $P(0,0)$, $P(0,1)$, $P(1,0)$, and $P(1,1)$.
We use this fact by extending the matrix interpretation of simple diagrams to one for decorated diagrams.

\subsection{Matrix interpretations}

All of the graphical languages considered in this paper come equiped with a complex matrix interpretation,
see Definition~\ref{defMatrixInterpretation}.
Moreover for each of these languages a diagram with $n$ inputs and $m$ outputs will be mapped to a matrix
with $\dim \Hilbert^{\tensor n}$ columns and $\dim \Hilbert^{\tensor m}$ rows,
where $\Hilbert$ represents $\bbC^2$.
It is this interpretation that allows us to use these languages to represent transformations in quantum computing.
Appendix \ref{secinterpretations} contains matrix interpretations for the ZX, ZH and ZW calculi,
or see \cite{ZH} and \cite{ZW} for the ZH and ZW interpretations, and
either \cite{Coecke08} or \cite{VilmartZX} for the ZX interpretation.

A set of diagrams parameterised over $\alpha$ is a set of simple diagrams,
and we could extend our interpretation
so that a set of simple diagrams is sent to a set of matrices.
This would, however, lose any structure from the phase algebra.
Instead we try to find a polynomial matrix interpretation;
for example one that sends a family of ZH equations $\set{\bbE}_{\alpha}$ to a matrix in $\Mat_{\bb{C}[\alpha]}$.

While this works well for ZH and \ZW, there is a complication with ZX:
A phase variable $\alpha$ in a ZX diagram corresponds to an $e^{i \alpha}$ in
the matrix interpretation.
We can try performing the substitution $Y:= e^{i \alpha}$,
but run into trouble if there is a node containing, for example, $-\alpha$
(and accordingly $Y\inv$ in the matrix,)
since polynomials do not normally allow negative powers.
Rather than stick with standard polynomials we instead move to Laurent polynomials;
polynomials that do allow positive and negative powers,
and define all the properties we will need of them below.

ZX also introduces one more sublety: 
There is an extra relation from the phase group ($2\pi = 0$) that we should take care to reflect in our matrix interpretation.
This does not impact universal ZX,
but does affect the fragments of ZX with a finite phase group
(Example \ref{excliffordt} demonstrates this for the Clifford+T fragment.)

\begin{definition}
A \textbf{complex Laurent polynomial interpretation} is a matrix interpretation:
\begin{align}
	\interpret{\cdot} 
	&: \bbL[\alpha_1, \dots, \alpha_n] \to \Mat_{\bbC[Y_1, Y_1\inv, \dots, Y_n, Y_n\inv]}
\end{align}

The source category here is the PROP of diagrams where generators are now taken not from the phase algebra of $\bbL$
but from this phase algebra adjoin the phase variables $\alpha_1 \dots \alpha_n$.

See appendix \ref{secinterpretations} for explicit interpretations of the ZX, ZW and ZH generators into
$\Mat_{\bbC[Y_1, Y_1\inv, \dots, Y_n, Y_n\inv]}$.
\end{definition}

\begin{example}
The Z spider from universal ZX is parameterised by an $\alpha \in [0, 2\pi)$,
and the (simple) matrix interpretation of some Z spider (with $\alpha$ instantiated at $a$) is:

\begin{align}
	\interpret{\family{\spider{gn}{\alpha}}\at{\alpha = a}} = 
\begin{bmatrix}
1 & 0 & \dots &  & 0 \\
0 & 0 & & & \vdots \\
\vdots & & \ddots & & \\
 & & & 0 & 0\\
0 & \dots & & 0 & e^{ia}
\end{bmatrix} \quad \in \Mat_{\bbC}
\end{align}

Rather than instantiate the value of $\alpha$ before we apply the map,
we instead make the substitution $Y = e^{i\alpha}$ to find a Laurent polynomial matrix interpretation:

\begin{align}
	\interpret{\spider{gn}{\alpha}} = 
\begin{bmatrix}
1 & 0 & \dots &  & 0 \\
0 & 0 & & & \vdots \\
\vdots & & \ddots & & \\
 & & & 0 & 0\\
0 & \dots & & 0 & Y
\end{bmatrix} \quad \in \Mat_{\bbC[Y, Y\inv]}
\end{align}
\end{example}

\subsection{Degree of a matrix}

\begin{definition} Laurent polynomial degrees (defined in the same manner as e.g. \cite{WolframLaurent}:)
\begin{itemize}
 \item The 0 polynomial has degree $- \infty$ by convention
\item the non-zero Laurent polynomial $a_n Y^n + a_{n-1} Y^{n-1} + \dots + a_0 + a_{-1} Y^{-1} + \dots + a_{-m} Y^{-m}$
with $a_n \neq 0$ and $a_{-m} \neq 0$
has \textbf{positive degree} $n \geq 0$ and \textbf{negative degree} $m \geq 0$.
\end{itemize}
Note that we can factorise this Laurent polynomial as $Y^{-m}$ multiplied by a (non-Laurent) polynomial.
\end{definition}

\begin{definition} We define here a notion of matrix and diagram degree:

\begin{itemize}
\item 
The \textbf{$Y_j^+$-degree} of a matrix in $\Mat_{\bbC[Y_1, Y_1\inv, \dots, Y_n, Y_n\inv]}$
is the maximum of the positive $Y_j$-degrees of the entries in that matrix
\item The \textbf{$Y_j^-$-degree} is the maximum of the negative $Y_j$-degrees of 
the entries in that matrix
\item 
The positive degree of a diagram is the positive degree of the matrix interpretation of that diagram
(likewise for negative degrees.)
When clear from context we will refer to the degree of a phase variable $\alpha_j$ in the diagram,
meaning the degree of $Y_j$ in the interpretation.
\end{itemize}

\end{definition}

\begin{example}
Here are example positive and negative degrees for first a polynomial, and then a $2\times2$ matrix of polynomials.
\begin{align}
\maxdegree{+}{Y}{(Y^8+1+Y^{-2})}& = 8 &
\maxdegree{-}{Y}{(Y^8+1+Y^{-2})}& = 2 \\
\nonumber \\
\maxdegree{+}{Y}{\begin{pmatrix}2 & Y^{-3} \\ Y & Y^2 - 2\end{pmatrix}}& = \max\set{0,0,1,2} = 2 &
\maxdegree{-}{Y}{\begin{pmatrix}2 & Y^{-3} \\ Y & Y^2 - 2\end{pmatrix}} & = \max\set{0,3,0,0} = 3
\end{align}
\end{example}

\begin{proposition} \label{propUpperBound}
For two diagrams $\bbD$ and $\bbD'$ we can find an upper bound for the degrees of the horizontal or vertical compositions of $\bbD$ and $\bbD'$,
i.e.:
\begin{align}
  \maxdegree{+}{}{
    (\bbD \comp \bbD')
  }
  &\leq \maxdegree{+}{Y}{\bbD} + \maxdegree{+}{Y}{\bbD'} &
  \maxdegree{+}{}{
    (\bbD \tensor \bbD')
  }
  &\leq \maxdegree{+}{Y}{\bbD} + \maxdegree{+}{Y}{\bbD'} \\
\nonumber \\
  \maxdegree{-}{}{
    (\bbD \comp \bbD')
  }
  &\leq \maxdegree{-}{Y}{\bbD} + \maxdegree{-}{Y}{\bbD'} &
  \maxdegree{-}{}{
    (\bbD \tensor \bbD')
  }
  &\leq \maxdegree{-}{Y}{\bbD} + \maxdegree{-}{Y}{\bbD'}
\end{align}
\end{proposition}

\begin{proof}
We first note that for Laurent polynomials $P$ and $P'$ in $\bbC[Y, Y\inv]$, and for $\lambda \in \bbC$:
\begin{align}
\maxdegree{+}{Y}{(\lambda P)} &\leq \maxdegree{+}{Y}{P} &
\maxdegree{-}{Y}{(\lambda P)} &\leq \maxdegree{-}{Y}{P}\\
\nonumber  \\
\maxdegree{+}{Y}{(P \times P')} & \leq \maxdegree{+}{Y}{P}  + \maxdegree{+}{Y}{P'}  &
\maxdegree{-}{Y}{(P \times P')} & \leq \maxdegree{-}{Y}{P}  + \maxdegree{-}{Y}{P'}  \\
\nonumber  \\
\maxdegree{+}{Y}{(\sum_j{P_j})} & \leq \max_j \maxdegree{+}{Y}{P_j} &
\maxdegree{-}{Y}{(\sum_j{P_j})} & \leq \max_j \maxdegree{-}{Y}{P_j} 
\end{align}

The composition $A \comp B$ or tensor product $A \tensor B$ of matrices produces a new matrix with entries that
are linear combinations of products of the entries of $A$ and $B$.
Therefore:

\begin{align}
\maxdegree{+}{Y}{(M \comp M')} & \leq \maxdegree{+}{Y}{M} + \maxdegree{+}{Y}{M'} &
\maxdegree{+}{Y}{(M \tensor M')} & \leq \maxdegree{+}{Y}{M} + \maxdegree{+}{Y}{M'} \\
\nonumber \\
\maxdegree{-}{Y}{(M \comp M')} & \leq \maxdegree{-}{Y}{M} + \maxdegree{-}{Y}{M'} &
\maxdegree{-}{Y}{(M \tensor M')} & \leq \maxdegree{-}{Y}{M} + \maxdegree{-}{Y}{M'}
\end{align}

Recalling that the degree of a diagram is the degree of its polynomial matrix interpretation
this gives us the result for $\bbD$ and $\bbD'$.
\end{proof}

\subsection{Finite verification}

\begin{restatable}{theorem}{thmparameter}For a diagrammatic equation with phase variables
\label{thmparameter}
\[ \family{\bb{\bbD}_1 = \bb{\bbD}_2}_{\alpha_1, \dots, \alpha_n}\]

that has a Laurent polynomial matrix interpretation,
and the equation is sound at all values of $(\alpha_1, \dots, \alpha_n) = (a_1, \dots, a_n) \in A_1 \times \dots \times A_n$,
where each $\abs{A_j}$ is sufficiently large,
then the equation is sound for all values of $(\alpha_1, \dots, \alpha_n)$. That is:

\begin{align}
& \interpret{\family{\bbD_1}_{\alpha_1, \dots, \alpha_n|\alpha_1 = a_1, \dots, \alpha_n = a_n}} = \nonumber
 \interpret{\family{\bbD_2}_{\alpha_1, \dots, \alpha_n|\alpha_1 = a_1, \dots, \alpha_n = a_n}} \\ 
 & \qquad \qquad \forall a_1 \in A_1, \dots, a_n \in A_n
\\ \nonumber
\implies & \interpret{\family{\bbD_1}_{\alpha_1, \dots, \alpha_n|\alpha_1 = a_1, \dots, \alpha_n = a_n}} =
 \interpret{\family{\bbD_2}_{\alpha_1, \dots, \alpha_n|\alpha_1 = a_1, \dots, \alpha_n = a_n}} \\ \nonumber
 & \qquad \qquad \forall a_1, \dots, a_n
\end{align} 

The necessary size of $\abs{A_j}$ is given by:
 
\begin{align}
  \abs{A_j} =& \max ( \maxdegree{+}{Y_j}{\bbD_1}, \maxdegree{+}{Y_j}{\bbD_2})
  + \max ( \maxdegree{-}{Y_j}{\bbD_1}, \maxdegree{-}{Y_j}{\bbD_2})
  + 1 
\end{align}

``The maximum of the positive degrees, plus the maximum of the negative degress, plus one.''
\end{restatable}

The proof can be found in appendix \ref{secproofparameter}. A sketch of it is as follows:
\begin{itemize}
	\item Manipulate the equation $\interpret{\bbD} = \interpret{\bbD'}$ into an equation of the form $M = 0$,
	where $M$ is a matrix of (non-Laurent) polynomials.
	\item Perform multivariate polynomial interpolation element-wise on $M$.
	\item In doing this interpolation we need to know the degrees of the polynomials in $M$,
	which we calculate using Proposition \ref{propUpperBound}
\end{itemize}

\begin{example}
Finding the sizes of the $A_j$: The following (universal) ZX diagram contains no !-boxes and two phase variables.
\begin{align}

\InputIfFileExists{param-example-zx}{}{\input{./tikz/param-example-zx.tikz}}
 & \qquad \label{eqnVerSpider}
\begin{matrix}
\deg^+_{\alpha_1}\bb{D}_1 = 1 &
\deg^+_{\alpha_2}\bb{D}_1 = 1 \\
\deg^-_{\alpha_1}\bb{D}_1 = 0 &
\deg^-_{\alpha_2}\bb{D}_1 = 0 \\
\end{matrix}
\qquad \begin{matrix}
\deg^+_{\alpha_1}\bb{D}_2 = 1 &
\deg^+_{\alpha_2}\bb{D}_2 = 1 \\
\deg^-_{\alpha_1}\bb{D}_2 = 0 &
\deg^-_{\alpha_2}\bb{D}_2 = 0 \\
\end{matrix}
\end{align}

We should therefore construct $A_1$ and $A_2$ such that
$\abs{A_1} =\max\set{1,1} + \max\set{0,0}+1$ and  $\abs{A_2} = \max\set{1,1} + \max\set{0,0}+1$.
By picking $A_1 = A_2 = \set{0 , \pi}$
we therefore know that we can verify this parameterised family of diagram equations for all values
of $\alpha_1$ and $\alpha_2$ by verifying this equation on the following grid of values:

\begin{align}
\begin{matrix}
& \alpha_1 = 0 & \alpha_1 = \pi \\
\alpha_2 = 0 & (0,0) & (0,\pi) \\
\alpha_2 = \pi & (\pi, 0) & (\pi, \pi)
\end{matrix}
\end{align}

i.e. by verifying the four equations:
\begin{align}
\set{\zxspider{0}{0}, \quad \zxspider{0}{\pi},\quad \zxspider{\pi}{0},\quad \zxspider{\pi}{\pi}}
\end{align}
we can assert that the diagram equation in \eqref{eqnVerSpider} is sound for all values of $\alpha_1$ and $\alpha_2$ in $[0, 2\pi)$.

\end{example}

\begin{corollary} \label{corZXPhases}
In the universal ZX calculus it suffices to check $(\alpha_1, \dots, \alpha_n) \in A_1 \times \dots \times A_n$
to prove an equation parameterised by $\alpha_j$, where the $A_j$ are sets of distinct angles
with 
\begin{align}
	\abs{A_j} = & \max
	\set{ \maxdegree{+}{\alpha_j}{\bbD_1},\maxdegree{+}{\alpha_j}{\bbD_2} }
	+ \max \set{ \maxdegree{-}{\alpha_j}{\bbD_1}, \maxdegree{-}{\alpha_j}{\bbD_2} }
	+1
\end{align}

\end{corollary}

\begin{remark} \label{remVilmart}
Corollary \ref{corZXPhases} was first proved in \cite[Theorem 3]{Vilmart18}.
The authors use the symbol $\mu$ to count appearances of $\alpha_j$ (with coefficient,)
and $T_j$ to denote a large enough set of values.
Their method does not use Laurent polynomials, instead examining ranks of certain matrices,
but this also means their method does not extend to ZH and ZW.

The result of \cite{Vilmart18} is in fact stronger than that of Theorem \ref{thmparameter};
as they show that under the conditions given here there is a \emph{universal} proof of the parameterised equation,
something that the method given in Theorem \ref{thmparameter} does not show.
For the purposes of verification, however, it is not important how an equation is derivable,
but only whether the equation is sound.
\end{remark}

\begin{example}

Note that the ZX result required the variables to be linear,
but for ZH and ZW this result applies to diagrams whose phases are polynomial
(or even Laurent polynomial) in the $\alpha_j$.
For example we can verify the following ZH equation by checking 3 distinct values of $\alpha$:
\begin{align}

\InputIfFileExists{ZH4}{}{\input{./tikz/ZH4.tikz}}

\end{align}

\end{example}

\begin{example}
\label{excliffordt}
Consider a Clifford+T ZX diagram that contains at least 8 nodes labeled by a positive $\alpha$.
Our theorem says that for any equation containing this diagram it suffices to try at least 9 distinct values of $\alpha$,
but this is impossible since there are only 8 distinct values of $\alpha$ available in Clifford+T.

This is because our Laurent polynomial matrix interpretation
needs to be viewed not in $\bb{C}[Y, Y\inv]$ but in $\bb{C}[Y] / (Y^8-1)$,
reflecting the property $8 \times \alpha = 0$ in our phase group.
All polynomials in $\bb{C}[Y] / (Y^8-1)$ have degree at most 7,
and so it is never necessary to check more than 8 points.
\end{example}

\section{Verifying !-boxes}

We now turn our attention to !-boxes.
The aim of this section is to show that
if an equation holds for $0 , \dots, N$ !-box instances,
then it continues to hold for any number of instances.
It then follows that one only needs to check the first $0, \dots, N$ instances
in order to verify the entire family.
The method relies on the finite dimensionality of what we call the \emph{join}
between the inside and outside of a !-box.
Sometimes this join is not finite dimensional,
and so we examine a property called \emph{separated}
which captures when the method will work.
Before we can get to that we first make clear how the nesting of !-boxes and phase variables can work.

\subsection{Children, copies, and the nesting order}
\label{secnesting}
We begin with some definitions for describing the effect of nesting !-boxes
inside a parameterised family of equations.
There is a choice to be made when nesting parameters:

\begin{definition} \label{defchild}
When a !-box creates new instances of a nested parameter
we \textbf{copy} that variable name,
so that all instances are linked by the same name (the approach taken in \cite[\S 4.4.2]{MerryThesis},)
we could alternatively create new names, each of which is referred to as a \textbf{child}
of the original parameter name (an approach requested by at least one user of Quantomatic.)
When we create child parameters we record the name of its parent,
so we can always tell the original ancestor of a parameter.
\end{definition}

Rather than pick one option over the other we will demonstrate our results for both choices.
In order to talk about nesting formally we introduce the following definition:

\begin{definition}

We define a partial order (which we call the \textbf{nesting order}) on !-boxes in a diagram:
\begin{align*}
\delta_1 < \delta_2 & \quad \text{ if $\delta_1$ is inside $\delta_2$} \\
\end{align*}

And use this partial order to draw a nesting diagram.
For example this (universal) ZX diagram:

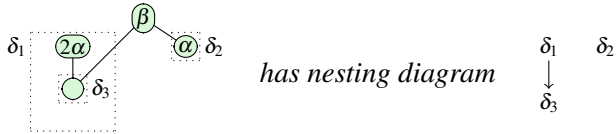
\begin{align*}

\InputIfFileExists{net-example-l}{}{\input{./tikz/net-example-l.tikz}}
 \quad \text{ has nesting diagram } \quad 
\InputIfFileExists{net-example-r}{}{\input{./tikz/net-example-r.tikz}}
 
\end{align*}

\end{definition}

\begin{definition}
We say an equation is \textbf{well nested} if the nesting diagrams corresponding to
the left and right hand sides of the equation are identical.
\end{definition}

\begin{definition}  \label{defJoin}
The \textbf{join} of a !-box is the collection of wires that leave that !-box,
i.e. the edges linking a vertex inside $\delta$ to one outside $\delta$.
The \textbf{size} of a join is the number of wires, and
the \textbf{dimension} of a join is the dimension of the diagram formed by just the wires of that join,
i.e. $\dim \interpret{ \;
\InputIfFileExists{wire}{}{\input{./tikz/wire.tikz}}
\;^{\tensor n} }$ where $n$ is the size of the join.
This is equivalently $(\dim \interpret{ \;
\InputIfFileExists{wire}{}{\input{./tikz/wire.tikz}}
\; })^n$.
\end{definition}

\begin{definition}  \label{defSeparated}

We describe a pair of !-boxes as \textbf{separated} if either:
\begin{itemize}
\item They are nested, or
\item There is no edge joining a vertex in one to a vertex in the other
\end{itemize}

Note that in many of the languages considered there are elementary operations one can do
that will change a diagram that is not separated into one that is.
See Section~\ref{secSeparability} for more details.

\end{definition}

\subsection{Verification}

\begin{definition}
Given a parameterised equation $\bbE$ we say that $\set{\bbE_1, \dots, \bbE_n}$ \textbf{verifies} $\bbE$ if:
\[
	\forall j \;(\bbE_j \text{ is sound}) \implies \bbE \issound
\]
\end{definition}

\begin{restatable}{theorem}{thmbboxx}\label{thmbbox2}
Given a family $\bbE = \family{\bb{D}_1 = \bb{D}_2}_{\alphadelta}$ of diagrammatic equations,
parameterised by a !-box $\delta_1$ 
where $\delta_1$ is separated from all other !-boxes
and is nested in no other !-box;
then $\bbE$ is verified by the finite family $\set{\bbE|_{\delta_1 = 0}, \dots, \bbE|_{\delta_1 = N}}$
where $N$ is the dimension of the join of $\delta_1$ in $\bb{D}_1$
plus the dimension of the join of $\delta_1$ in $\bb{D}_2$.
That is:
\begin{align}
n_1 :=& \text{ join of $\delta_1$ in $\bbD_1$} \nonumber \\
n_2 :=& \text{ join of $\delta_1$ in $\bbD_2$} \nonumber \\
N :=& \dim \left( \Hilbert^{\tensor n_1} \oplus \Hilbert^{\tensor n_2} \right) \nonumber \\[2em]
& \interpret{\family{\bb{D}_1}_{\alphadelta | \alpha_1 = a_1 , \dots, \delta_1 = d_1, \dots \delta_m = d_m}} \nonumber \\
= &
\interpret{\family{\bb{D}_2}_{\alphadelta | \alpha_1 = a_1 , \dots, \delta_1 = d_1, \dots \delta_m = d_m}} 
& \forall d_1 \leq N \\[2em]
\implies
& \interpret{\family{\bb{D}_1}_{\alphadelta | \alpha_1 = a_1 , \dots, \delta_1 = d_1, \dots \delta_m = d_m}}  \nonumber \\
= &
\interpret{\family{\bb{D}_2}_{\alphadelta | \alpha_1 = a_1 , \dots, \delta_1 = d_1, \dots \delta_m = d_m}} 
& \forall d_1
\end{align}
\end{restatable}

The proof is presented in appendix \ref{secproofbbox2}.
The essence of the proof is to change the presentation of the !-box instances
into one that looks like the vertical composition of repeated elements.
From there some elementary (albeit fiddly) properties of finite dimensional vector spaces are applied
to achieve the result.
Note that the the resulting finite family of equations need not be simple,
and that if $\delta_1$ were not separated from the other !-boxes then we could not put a bound on $N$.

\begin{example}
Consider the following (universal) ZX family of equations, parameterised by a single \mbox{!-box}:
\begin{align}
\bbE : =\family{
\InputIfFileExists{internal-rule}{}{\input{./tikz/internal-rule.tikz}}
}_{\delta}
\end{align}
The join between the !-box and the rest of the diagram on the left is two wires (dimension $2^2$),
and on the right is one wire (dimension $2^1$.)
These sum to have dimension 6, and we therefore need only to check the !-box instances
($0, \dots, 6$)
to be sure that the equation $\bbE$ is sound.
\end{example}

\section{Interacting Parameters}
\label{sectogether}

Theorem \ref{thmparameter} and theorem \ref{thmbbox2} deal with equations containing
multiple phase variables and nested !-boxes respectively.
This section will put together the necessary results such that we can combine these approaches to deal with
equations containing multiple !-boxes and phase variables, any of which could potentially be nested inside other !-boxes.

\begin{theorem} \label{thmboth}
Given an equation\ $\bbE$\ and ways of finding
\begin{itemize}
	\item finite verifying sets $A_j$ for the $\alpha_j$ (Theorem~\ref{thmparameter})
	\item finite verifying sets $D_k$ for the $\delta_k$ (Theorem~\ref{thmbbox2})
\end{itemize}
then we may verify the entire family $\bbE$ by checking the (finite)
set given by the cartesian product of all the $A_j$ and $D_k$.
\end{theorem}

\begin{proof}
We define:
\begin{align}
\bar D & := D_1 \times D_2  \times \dots \times D_m  \\
\bar \delta & := (\delta_1, \dots, \delta_m)  \\
A_j(\bar \delta) & := \text{The verifying set for } \alpha_j
\text { once } E \text{ has had !-boxes instantiated at } \bar \delta
\end{align}
Construct $A_j$ by choosing as many points as there are in $\max_{\bar \delta}\set{ \abs{A_j(\bar \delta)}}$.
This is finite because the $D_i$ and the $A_j(\bar \delta)$ are finite.
$A_j$ therefore contains enough points to be a verifying set for $A_j(\bar \delta) \quad \forall \bar \delta \in \bar D$.
We now show that $A_1 \times \dots \times A_n \times D_1 \times \dots \times D_m$
describes a verifying set for the parameterised equation $\bbE$:
\begin{align}
&\bbE_{\alphadelta | \forall i, j \; \alpha_j \in A_j ,\, \delta_k \in D_k} & \issound\\
= \quad & \bbE_{\alphadelta | \forall i \; \alpha_j \in A_j ,\, \bar \delta \in \bar D} & \issound
& \quad \text{(rewrite using }\bar \delta \text{ notation})\\
\implies &  \bbE_{\alphadelta | \forall i \; \alpha_j \in A_j(\bar \delta), \, \bar \delta \in \bar D} & \issound &\quad  \text{(construction of $A_j$)}  \\
\implies &  \bbE_{\alphadelta | \bar \delta \in \bar D}  & \issound & \quad \text{(theorem \ref{thmparameter})} \\
\implies & \bbE_{\alphadelta}  & \issound &\quad \text{(theorem \ref{thmbbox2})} 
\end{align}

\end{proof}

\begin{restatable}{theorem}{thmFinite}\label{thmfinite}
Given a parameterised family of equations $\bbE$,
where the !-boxes are separated and well nested,
we can construct a finite set of simple equations $\set{\bbE_\kappa}_{\kappa \in K}$,
such that:
\begin{align}
\set{\bbE_\kappa}_{\kappa \in K} \quad \issound & \implies \bbE \quad \issound
\end{align}

\end{restatable}

The proof can be found in appendix \ref{secfinite}.
The idea of the proof is to iteratively remove dependencies on !-boxes via theorem \ref{thmbbox2},
each time generating a larger set of verifying equations.
Once we have removed all !-box dependence we then use theorem \ref{thmboth}
to remove phase variable dependence;
using the ``largest'' equation in $\set{\bbE_\kappa}$ to determine the sizes of the $A_j$.
We argue by the finite nature of all the parts involved that this process terminates.

We did not specify in the statement of our theorem which method of !-box expansion we were following
(see Definition~\ref{defchild})
and indeed both methods work and are covered in the proof.

\section{Conclusions}

We have shown how to construct finite sets of equations that verify certain classes
of infinite families of equations.
To our knowledge this is the first algorithm that can verify equations containing !-boxes,
and the first analysis of such algorithms for ZW and ZH.
This paves the way for proof assistants to verify parameterised theorems.
Further work would include combining this work with conjecture synthesis,
so that a computer could generate and verify parameterised hypotheses.
One could also implement such methods into proof assistants,
such as Quantomatic \cite{Quantomatic},
so that the verifying set could be generated (and ideally checked) automatically.

One final avenue is to develop methods to deal with the side conditions of the (EU) rule of the ZX calculus (\cite{VilmartZX},)
either by extending these results or finding a presentation of the ZX calculus that does not require side conditions.
It should be noted that \cite{Vilmart18} gives some evidence that such a presentation may not exist.

\subsection{Acknowledgements}

The author would like to thank
their supervisors Aleks Kissinger and Bob Coecke as well as Miriam Backens for their help in writing this paper,
and the EPSRC for providing the funding grant.

\bibliographystyle{eptcs}
\bibliography{references}

\appendix

\newpage

\section{Proofs}

We show here any proofs that were not included in the main text.

\subsection{Proof of theorem \ref{thmparameter}}

\thmparameter*

\label{secproofparameter}

\begin{proof}

We are seeking the multivariate complex polynomials that populate the matrices $\interpret{\bb{D}_1}$ and $\interpret{\bb{D}_2}$.
We begin by combining the two matrices of Laurent polynomials into one matrix of (not-Laurent) polynomials and a scale factor of the form $Y_1^{m_1} \dots Y_n^{m_n}$.

\begin{itemize}

\item We define:
\begin{align}
M_1 :=  \interpret{\bb{D}_1} \quad \quad
M_2 :=  \interpret{\bb{D}_2} 
\end{align}
and wish to show $M_1 = M_2$.

\item First we pull enough copies of $Y_1\inv , \dots, Y_n\inv$ out of each side so that we have an equation of the form:
\begin{align}
 M_1' \prod_j (Y_j\inv)^{\maxdegree{-}{Y_j}{M_1}} =  M_2'  \prod_j (Y_j\inv)^{\maxdegree{-}{Y_j}{M_2}}
\end{align}
Where $M_1'$ and $M_2'$ are matrices of (not-Laurent) polynomials.

\item Let $m_j := \max(\maxdegree{-}{Y_j}{M_1}, \maxdegree{-}{Y_j}{M_2})$
and multiply both sides by $\prod_j Y_j^{m_j}$ to clear any negative powers of $Y_j$.
\begin{align}
	M_1' \prod_j Y^{m_j - \maxdegree{-}{Y_j}{M_1}}  =  M_2' \prod_j Y^{m_j - \maxdegree{-}{Y_j}{M_2}}
\end{align}

\item Then subtract the right hand side from the left:
\begin{align}
	&M_1' \prod_j Y_j^{m_j - \maxdegree{-}{Y_j}{M_1}}  -  M_2' \prod_j Y^{m_j - \maxdegree{-}{Y_j}{M_2}} = 0  \\
	\bb{M} :=&  M_1' \prod_j Y^{m_j - \maxdegree{-}{Y_j}{M_1}}  -  M_2' \prod_j Y^{m_j - \maxdegree{-}{Y_j}{M_2}}
\end{align}

Note that $\bb{M}$ is a matrix of (not-Laurent) polynomials.
The statement $\bb{M} = 0$ can be viewed as a stating that each entry of $\bb{M}$ is equal to the 0 polynomial.

\item We will use the notation $\maxdegree{}{Y_j}{\cdot}$ for the degree of a (not-Laurent) polynomial,
or matrix of polynomials.
We could continue to use the term positive degree, the definitions coincide, but want to make it clear that
we are no longer in a Laurent polynomial setting.

\item We wish to find a bound for the maximum degree of any polynomial in $\bb{M}$:
\begin{align}
 \maxdegree{}{Y_j}{\bb{M}} = 
 & \maxdegree{}{Y_j}{ M_1' \prod_j Y^{m_j - \maxdegree{-}{Y_j}{M_1}}  -  M_2' \prod_j Y^{m_j - \maxdegree{-}{Y_j}{M_2}}} \\
 \leq & \max ( \maxdegree{}{Y_j}{M_1'} + m_j - \maxdegree{-}{Y_j}{M_1},
  \maxdegree{}{Y_j}{M_2'} + m_j - \maxdegree{-}{Y_j}{M_2} ) \\
 = & \max ( \maxdegree{+}{Y_j}{M_1} + \maxdegree{-}{Y_j}{M_1} + m_j - \maxdegree{-}{Y_j}{M_1},
  \nonumber \\ & \qquad \maxdegree{+}{Y_j}{M_2} +  \maxdegree{-}{Y_j}{M_2} + m_j - \maxdegree{-}{Y_j}{M_2} )\\
= & \max ( \maxdegree{+}{Y_j}{M_1} + m_j  , \maxdegree{+}{Y_j}{M_2} + m_j )\\
 = & \max ( \maxdegree{+}{Y_j}{M_1}, \maxdegree{+}{Y_j}{M_2}) + m_j\\
\label{eqnparam-deg} = & \max ( \maxdegree{+}{Y_j}{M_1}, \maxdegree{+}{Y_j}{M_2}) +
 \max ( \maxdegree{-}{Y_j}{M_1}, \maxdegree{-}{Y_j}{M_2})
\end{align}

\item Suppose we know that our diagram equation
is sound for parameter choices in a large enough \emph{regular grid} of values.
(A technique that appears to date to before the 20th century, according to \cite{multivariate}.)
\begin{align}
&\interpret{\family{\bb{D}_1}_{\alpha_1, \dots, \alpha_n|\alpha_1 = a_1, \dots, \alpha_n=a_n}} =
\interpret{\family{\bb{D}_2}_{\alpha_1, \dots, \alpha_n|\alpha_1 = a_1, \dots, \alpha_n=a_n}}  \\
\nonumber\text{for}\quad& (a_1, \dots, a_n) \in A_1 \times \dots \times A_n &\\
\nonumber\text{where}\quad & \abs{A_j} = \deg_{Y_j}(\bb{M}) +1 &\\
\nonumber A_j \quad & := \set{a_{j,0}\;, \dots,\; a_{j,\deg{Y_j}}}
\end{align}

By picking a polynomial entry $P$ of $\bb{M}$, expressing $P$ using the multi-index $\beta$ as
$P = \sum_{\beta} c_\beta Y_j^\beta$,
and then evaluating $P$ at every point in $A_1 \times \dots \times A_n$
we construct the system of equations:
\begin{align}
\begin{bmatrix}
c_{0, \dots,0}a^{0, \dots,0}_{0, \dots, 0} & c_{1, \dots,0}a^{1, \dots,0}_{0, \dots, 0} &  \dots &c_{\deg{Y_1}, \dots, \deg{Y_n}}a^{\deg{Y_1}, \dots,\deg{Y_n}}_{0, \dots, 0} \\
\vdots & \vdots & \ddots & \vdots \\
c_{0, \dots,0}a^{0, \dots,0}_{\abs{A_1}, \dots, \abs{A_n}} & c_{1, \dots,0}a^{1, \dots,0}_{\abs{A_1}, \dots, \abs{A_n}} &  \dots &c_{\deg{Y_1}, \dots, \deg{Y_n}}a^{\deg{Y_1}, \dots,\deg{Y_n}}_{\abs{A_1}, \dots, \abs{A_n}} \\
\end{bmatrix}
=
\begin{bmatrix}
0 \\ \vdots \\ 0
\end{bmatrix}
\end{align}

Which we view as:
\begin{align}
\begin{bmatrix}
V
\end{bmatrix}
\begin{bmatrix}
c_{0, \dots, 0} \\ \vdots \\ c_{\deg{Y_1}, \dots, \deg{Y_n}}
\end{bmatrix}
=
\begin{bmatrix}
0 \\ \vdots \\ 0
\end{bmatrix}
\end{align}

Where $V$ contains all the products $a_1^{\beta_1} \times \dots \times a_n^{\beta_n}$,
 $\beta$ ranging from $(0, \dots, 0)$ to $(\deg{Y_1}, \dots, \deg{Y_n})$.
Thankfully $V$ decomposes as:

\begin{align}
V = & V_1 \tensor \dots \tensor V_n \\
V_j = &
\begin{bmatrix}
a_{j,0}^0 & a_{j,0}^1 & \dots & a_{j,0}^{\deg{Y_j}} \\
a_{j,1}^0 & a_{j,1}^1 & \dots & a_{j,1}^{\deg{Y_j}} \\
\vdots & \vdots & & \vdots \\
a_{j,\deg{Y_j}}^0 & a_{j,\deg{Y_j}}^1 & \dots & a_{j,\deg{Y_j}}^{\deg{Y_j}} \\
\end{bmatrix}
\end{align}

\item Since $\det(A \tensor B) \neq 0$ if and only if $\det(A) \neq 0$ and $\det(B) \neq 0$,
and since $\det(V_j) \neq 0$ because each $V_j$ is a Vandermonde matrix,
we know that $\det(V) \neq 0$.
Since $V$ is therefore invertible we know that all the coefficiencts $c_\beta$ must be 0,
and therefore $P$ is the 0 polynomial.

\item In the presence of a regular grid on which $\bb{D}_1$ and $\bb{D}_2$ agree we know:

\begin{align}
&&\interpret{\family{\bb{D}_1}_{\alpha_1, \dots, \alpha_n|\alpha_1 \in A_1, \dots, \alpha_n \in A_n}}
&= \interpret{\family{\bb{D}_2}_{\alpha_1, \dots, \alpha_n|\alpha_1 \in A_1, \dots, \alpha_n \in A_n}} \\
&\implies&P &= 0 \quad \text{ for any entry } P \text{ of } M \\
&\implies&\bb{M} &= 0 \\
&\implies& M_1' \prod_j Y^{m_j - \maxdegree{-}{Y_j}{M_1}}  -  M_2' \prod_j Y^{m_j - \maxdegree{-}{Y_j}{M_2}} &=0 \\
&\implies&  M_1' \prod_j Y^{m_j - \maxdegree{-}{Y_j}{M_1}}  &= M_2' \prod_j Y^{m_j - \maxdegree{-}{Y_j}{M_2}}  \\
&\implies& M_1 &= M_2\\
&\implies&\interpret{\family{\bb{D}_1}_{\alpha_1, \dots, \alpha_n}} &= \interpret{\family{\bb{D}_2}_{\alpha_1, \dots, \alpha_n}}
\end{align}

\item By setting $d_j$ to $\maxdegree{}{Y_j}{\bb{M}} + 1$
(which we can calculate using equation (\ref{eqnparam-deg})), and $\abs{A_j} = d_j$ we attain our result;
that if we know that the interpretations of the families of diagrams agree on the regular grid described by the sizes $d_j$
then the interpretations agree on all points in the phase group.

\end{itemize}
\end{proof}

\subsection{Proof of theorem \ref{thmbbox2}}
\label{secproofbbox2}

\thmbboxx*

The idea of the proof is:
\begin{enumerate}
\item Manipulate the diagrams into what we call \textbf{series !-box form}
\item Move to the matrix interpretation
\item Manipulate the equation between two matrices into an expression on a single vector space of dimension $N$
\item Demonstrate the required property as a condition on subspaces
\end{enumerate}

We will require the ``only topology matters'' meta-rule for our diagrams,
suitable spider laws, a matrix interpretation,
and finite dimensionality of $\Hilbert$.

\begin{proof}

We begin by showing series !-box form on a diagram containing a single !-box:

\begin{align} \family{\bb{D}}_{\delta} \end{align}

First we will manipulate the diagram (via ``only topology matters'') until it is in the following form:

\begin{align} \bb{D} \quad = \quad  
\InputIfFileExists{d_and_b}{}{\input{./tikz/d_and_b.tikz}}
 \end{align}

We note that the nodes inside $G$ that join with $B$ must be spiders;
since there can be arbitrary many instances of $B$ the connecting node in $G$ must be able to have arbitrary arity.
There may also be $p$ boundaries that are internal to $B$ (and therefore $\delta$,)
which we will deal with momentarily.
We will now rely on the existence of a spider law such that we may do the following:
\begin{align}

\InputIfFileExists{gen-spider-l}{}{\input{./tikz/gen-spider-l.tikz}}
 \quad = \quad 
\InputIfFileExists{gen-spider-r}{}{\input{./tikz/gen-spider-r.tikz}}

	\end{align}

Which is the ability to ``spread'' a spider with $n$ outputs into $n$
repeated copies of a spider with 1 output, with suitable initial, terminal, and joining subdiagrams.
We do this so that each instance of the !-box is connected to its own copy of the spider,
and these copies are joined in sequence.

This is possible in ZX, ZH and ZW.
To give a ZX example (where the spider law is simple:)

\begin{align}
\family{
\InputIfFileExists{internal-spider-join}{}{\input{./tikz/internal-spider-join.tikz}}
}_{\delta|\delta=d}\quad = \quad 
\InputIfFileExists{bbox-sideways}{}{\input{./tikz/bbox-sideways.tikz}}

\end{align}

Note that although we have only used one example node and joining wire we can perform this action
on all nodes and joining wires.
Where two wires travel from the !-box to the same spider inside $G$
we first spread out that spider so each wire from the !-box connects to a different spider in $G$.
Here is an example for $n$ wires between $B$ and $G$, and $p$ boundaries inside $\delta$
(which we stretch down to be below each copy of $B$ in this representation, just for visibility.)

\begin{align}\label{eqnseries-bbox}

\InputIfFileExists{bbox-sideways-compact-external}{}{\input{./tikz/bbox-sideways-compact-external.tikz}}

\end{align}

From here it is easy to see that we have the diagram $G : m \to n$,
beside $d$ copies of a diagram we call $B : p + n \to n$ (containing the $p$ boundary nodes and the new connecting spiders),
and finally an ending diagram $C: n \to 0$.
We call this the \textbf{series !-box form}.

\begin{definition}
\textbf{Series !-box form} for a given (non-nested, separated) !-box $\delta_1 = d$ in a diagram
is a presentation (as in equation \ref{eqnseries-bbox}) of each the $\delta_1$-instantiated diagrams as
\begin{align*}
C &:= \text{ the end cap of spiders} \\
B &:= \text{ repeated element, which may contain } \alpha_1, \dots, \alpha_n, \delta_2, \dots, \delta_m,
 \text{ and some boundary nodes} \\
G &:= \text{ the rest of the diagram outsie of B, which may contain } \alpha_1, \dots, \alpha_n, \delta_2, \dots, \delta_m,
\\& \qquad \text{ and some boundary nodes}\\
\end{align*}
Such that the $d$ instances of $\delta_1$ are spread out as $d$ instances of $B$.
In the case where you consider a !-box to create child instances of parameters
then $B$ will contain children of the $\alpha_j$ and $\delta_{k>1}$,
rather than copies.
\end{definition}

\textbf{Claim: \;} In the languages ZX, ZW and ZH, 
for any diagram $\bbD$ and any value of $d$ we can put $\bbD \at{\alphadelta | \delta_1=d}$ into series !-box form
as in the example above.

Proof of claim: Inspection of the spider rules in each language.

\vspace{2em}

We will just use the variable names $\alphadelta$ here
and assume we are copying parameters,
but the technique is identical for when one is creating child parameters
and we will point out the different intricacies along the way.
We aim to show

\begin{align} \label{eqnbbox2}
\nonumber & \interpret{\family{\bb{D}_1}_{\alphadelta | \alpha_1 = a_1 , \dots, \delta_1 = d_1, \dots \delta_m = d_m}} \\
\nonumber = &
\nonumber \interpret{\family{\bb{D}_2}_{\alphadelta | \alpha_1 = a_1 , \dots, \delta_1 = d_1, \dots \delta_m = d_m}} 
& \forall d_1 \leq N \\ \nonumber \\
\implies
& \interpret{\family{\bb{D}_1}_{\alphadelta | \alpha_1 = a_1 , \dots, \delta_1 = d_1, \dots \delta_m = d_m}} \\
\nonumber = &
\interpret{\family{\bb{D}_2}_{\alphadelta | \alpha_1 = a_1 , \dots, \delta_1 = d_1, \dots \delta_m = d_m}} 
& \forall d_1 \\ \nonumber \\
\nonumber N :=\; & \dim \left( \Hilbert^{\tensor n_1} \oplus \Hilbert^{\tensor n_2} \right) \\
\nonumber n_1 :=\; & \text{ join of $\delta_1$ in $\bb{D}_1$} \\
\nonumber n_2 :=\; & \text{ join of $\delta_1$ in $\bb{D}_2$}
\end{align}

We would like to move directly to the matrix interpretation of diagram (\ref{eqnseries-bbox}),
but we have the following problems:
\begin{itemize}
\item the $p$ dangling wires from every copy of $B$
\item the parameters inside every copy of $B$ (either linked copies or discrete children)
\end{itemize}

We solve these problems (and justify these solutions below) by considering $B$ as parameterised by:
\begin{itemize}
	\item Instances of $\alpha_j \; \forall j$
	\item Instances of $\delta_k \; \forall k > 1$
	\item Input vectors $v \in \Hilbert^{\tensor p}$ that ``plug'' the inputs inside $B$.
\end{itemize}

Since we can show equivalence of complex matrices by showing that they perform the same operation on any input,
we need to show that for any choice of $\alpha_j$, $\delta_{k>1}$
and for any input that equation (\ref{eqnbbox2}) holds.
(And that if we are creating child instances of parameters then these equations hold
for any choice of each of those independently.)
Once we have specified values of $\alpha_j$, $\delta_{k>1}$ we may use our matrix interpretation
to get a complex matrix, but we need to do this for every possible choice of  $\alpha_j$, $\delta_{k>1}$.

Assuming we have made choices for the $\alpha_j$ and $\delta_{k>1}$,
we wish to justify that we can choose the input vector for the $p$ inputs of $B$ independently.
To do this we note that we can determine equality of complex matrices by showing they
perform the same operation on all basis elements.

\vskip 2em
\textbf{Claim: \;} The set of all vectors of the form
$\set{v_d \tensor \dots \tensor v_1 \tensor x}$
where
$v_j \in \Hilbert^{\tensor p}$
and
$x \in \Hilbert^{\tensor m}$ contains a basis for
$\Hilbert^{\tensor(m + dp)}$
.

Proof of claim: Note that we may form a basis for $V \tensor V'$ by taking the tensor products of the bases of $V$ and $V'$,
and therefore the above set contains all the basis elements of $\Hilbert^{\tensor p} \tensor \dots \Hilbert^{\tensor p} \tensor
\Hilbert^{\tensor m} \iso \Hilbert^{\tensor (m +dp)}$.

\vskip 2em

Given a choice of values for the $\alpha_j$, $\delta_{k > 1}$ and $v$
we denote this choice by $q$ and use $B_q$ to mean ``the sub-diagram $B$ from the series !-box form with this choice of variables''.
If one is copying variable names then these values must be the same in each copy of $B_q$,
but we will show the more general case of when you cannot assume that each instance of $B_q$ contains the same
choices of values for variables.
(Even when creating children of the $\alpha_j$ 
and $\delta_{k>1}$ these children are contained entirely inside each
instance of $B$ so specifying a choice for $q$ for each $B$ determines all
the values for parameters inside $B$.)

The last thing to define here is $G_q$.
$G_q$ is the subdiagram $G$ instantiated with the values described by $q$, same as for $B_q$.
Once we have chosen values for $q$ we may consider the matrix interpretation of the diagram:

\begin{align}
&\interpret{C} \interpret{B_{q_d}} \dots \interpret{B_{q_1}} \interpret{G_{q_o}} \\
G_q & : \Hilbert^{\tensor m} \to  \Hilbert^{\tensor n} \\
B_q &:  \Hilbert^{\tensor n} \to  \Hilbert^{\tensor n} \\
C &:  \Hilbert^{\tensor n} \to \bbC
\end{align}

Given an equation $\bbD_1 = \bbD_2$ of two families of diagrams, both parameterised by a (non-nested, separated) !-box $\delta_1$
(among other parameters)
we wish to remove our dependence on $\delta_1$ by
instead verifying a finite set of equations, each of which has a different value for $\delta_1$.
Note that for this to be the case we require the number of inputs to be equal;
i.e. $m_1 = m_2 =m$ and $p_1 = p_2 =p$, but we do not require $n_1 = n_2$ in equation \ref{eqnseries-bbox}.
With the aim of reducing notational clutter we instantiate $\delta_1 = d$ and express $D_1$ and $D_2$ in series !-box form,
with matrix interpretations:

\begin{align}
&\interpret{C_1} \interpret{B_{1,q_d}} \dots \interpret{B_{1,q_1}} \interpret{G_{1,q_0}} \\ 
&\interpret{C_2} \interpret{B_{2,q_d}} \dots \interpret{B_{2,q_1}} \interpret{G_{2,q_0}}
\end{align}

And we wish to know when these two interpretations are equal.
Rather than consider the matrices acting on two independent spaces
we view them as acting on the direct sum of those two spaces
and represent these maps as block matrices. 
(We drop the $\interpret{\cdot}$ notation when it would appear inside a matrix.)

\begin{align}
&\interpret{C_1} \interpret{B_{1,q_d}} \dots \interpret{B_{1,q_1}} \interpret{G_{1,q_0}} \nonumber \\
= &\interpret{C_2} \interpret{B_{2,q_d}} \dots \interpret{B_{2,q_1}} \interpret{G_{2,q_0}} & \forall d, q \\
\iff&
	\begin{bmatrix}	1 & -1\end{bmatrix}
	\begin{bmatrix}C_1 & 0 \\ 0 & C_2\end{bmatrix}
	\begin{bmatrix}B_{1, q_d} & 0 \\ 0 & B_{2, q_d}\end{bmatrix}
	\dots
	\begin{bmatrix}B_{1, q_1} & 0 \\ 0 & B_{2, q_1}\end{bmatrix}
	\begin{bmatrix}G_{1,q_0} & 0 \\ 0 & G_{2,q_0}\end{bmatrix}
	\begin{bmatrix}	\id_m \\ \id_m\end{bmatrix}
     = 0 & \forall d, q
\end{align}

Think of this as copying an input vector $x \in \Hilbert^m$ as $x :: x$ in $\Hilbert^m \oplus \Hilbert^m$,
then applying
\begin{align*}
&\interpret{C_1} \interpret{B_{1,q_d}} \dots \interpret{B_{1,q_1}} \interpret{G_{1,q_0}} \\ 
\text{and} \\
&\interpret{C_2} \interpret{B_{2,q_d}} \dots \interpret{B_{2,q_1}} \interpret{G_{2,q_0}}
\end{align*}
to the left and right copies respectively.
After that we apply a minus sign to the right hand result and add that to the left hand result,
effectively comparing them and demanding the difference to be 0.
We seek to prove:

\begin{align}
	& \nonumber
	\begin{bmatrix}	1 & -1\end{bmatrix}
	\begin{bmatrix}C_1 & 0 \\ 0 & C_2\end{bmatrix}
	\begin{bmatrix}B_{1, q_d} & 0 \\ 0 & B_{2, q_d}\end{bmatrix}
	\dots
	\begin{bmatrix}B_{1, q_1} & 0 \\ 0 & B_{2, q_1}\end{bmatrix}
	\begin{bmatrix}G_{1,q_0} & 0 \\ 0 & G_{2,q_0}\end{bmatrix}
	\begin{bmatrix}	\id_m \\ \id_m\end{bmatrix}
     = 0 & \forall d \leq N, q \\
	\implies 
	&
	\begin{bmatrix}	1 & -1\end{bmatrix}
	\begin{bmatrix}C_1 & 0 \\ 0 & C_2\end{bmatrix}
	\begin{bmatrix}B_{1, q_d} & 0 \\ 0 & B_{2, q_d}\end{bmatrix}
	\dots
	\begin{bmatrix}B_{1, q_1} & 0 \\ 0 & B_{2, q_1}\end{bmatrix}
	\begin{bmatrix}G_{1,q_0} & 0 \\ 0 & G_{2,q_0}\end{bmatrix}
	\begin{bmatrix}	\id_m \\ \id_m\end{bmatrix}
     = 0 & \forall d, q
\end{align}

Recalling that $q$ is the choice of values for $\alpha_j$, $\delta_{k>1}$ and $v \in \Hilbert^{\tensor p}$,
we use $Q$ to denote the set of all possible choices.
We use $B'_q$ for the matrix that acts as the direct sum of $B_{1,q}$ and $B_{2,q}$:

\begin{align}
	B'_q := &
	\begin{bmatrix}B_{1, q} & 0 \\ 0 & B_{2, q}\end{bmatrix}
\end{align}

and inductively define the spaces:

\begin{align}
V_0 & := \text{span} \set{ \bigcup_{q \in Q} \text{Im} (\begin{bmatrix}G_{1,q} & 0 \\ 0 & G_{2,q}\end{bmatrix}\begin{bmatrix}\id_m \\ \id_m\end{bmatrix})} \\
V_j & := \text{span} \set{ V_{j-1} \cup \; \bigcup_{q \in Q} B'_q V_{j-1}}
\end{align}

The $V_j$ form an ascending sequence of subspaces, each containing the potential images of up to
$j$ applications of $B'_q$:
\begin{align}
	V_j \geq \text{Im}(\;B'_{q_k} \dots B'_{q_1}\, V_0\;)  \quad \forall k \leq j\; \forall q_k, \dots, q_1 \in Q
\end{align}

\vskip 2em

\textbf{Claim:} \;
There is a number $b$ such that
\begin{itemize}
\item if $j <b$ then $V_j > V_{j-1}$
\item if $j \geq b$ then $V_j = V_{j-1}$
\item $b \leq \dim ( \; \Hilbert^{n_1} \oplus \Hilbert^{n_2} \; )$
\end{itemize}

Proof of claim: Define $N:=\dim ( \; \Hilbert^{n_1} \oplus \Hilbert^{n_2} \; )$ and then
\begin{itemize}
\item $V_{j-1} \leq V_j \quad \forall j$
\item if $V_j = V_{j-1}$ then $V_{j+1} = V_j$
\item $V_j \leq \Hilbert^{n_1} \oplus \Hilbert^{n_2} \quad \forall j$
\item $\dim V_{j-1} \leq \dim V_j \leq N \quad \forall j$
\item The strictly increasing section of the sequence of the $\dim V_j$ must have length less than $N$
\item We declare $b$ to be the number such that $V_c = V_b \; \forall c \geq b$, and note $b \leq N$
\end{itemize}

\vskip 2em

Let $W$ be the kernal of the map $
	\begin{bmatrix}	1 & -1\end{bmatrix}
	\begin{bmatrix}C_1 & 0 \\ 0 & C_2\end{bmatrix}$

\textbf{Claim:} \;
\begin{align}
&V_j \leq W & \forall j \leq N\\ \nonumber
\quad\quad\quad\implies & V_j \leq W &\forall j
\end{align}

Proof of claim: Since $V_c = V_N $ when $ c \geq N \geq b$ it is enough to show that this is the case for all $V_j$ when $j \leq N$.
This is implied by the assumption in our theorem;
that for $d \leq N$ our diagrammatic equation holds,
and so for any choice of $d$ and $q_1, \dots, q_d$ this matrix equation holds:

\begin{align}
	\begin{bmatrix}	1 & -1\end{bmatrix}
	\begin{bmatrix}C_1 & 0 \\ 0 & C_2\end{bmatrix}
	\begin{bmatrix}B_{1, q_d} & 0 \\ 0 & B_{2, q_d}\end{bmatrix}
	\dots
	\begin{bmatrix}B_{1, q_1} & 0 \\ 0 & B_{2, q_1}\end{bmatrix}
	\begin{bmatrix}G_{1,q_0} & 0 \\ 0 & G_{2,q_0}\end{bmatrix}
	\begin{bmatrix}	\id_m \\ \id_m\end{bmatrix}
     = 0
\end{align}

\vskip 2em

We have shown that for any choice of inputs $v \in \Hilbert^p \tensor \dots \tensor \Hilbert^p \tensor \Hilbert^m$,
 and parameters $\alpha_j$, $\delta_{k>1}$ our matrix equations hold,
and by extension they hold on all elements of the space $\Hilbert^{\tensor (m + dp)}$,
i.e. that:

\begin{align}
&\forall d \leq N \; \forall q_d \dots q_1 \nonumber& \\
&	\begin{bmatrix}	1 & -1\end{bmatrix}
	\begin{bmatrix}C_1 & 0 \\ 0 & C_2\end{bmatrix}
	\begin{bmatrix}B_{1, q_d} & 0 \\ 0 & B_{2, q_d}\end{bmatrix}
	\dots
	\begin{bmatrix}B_{1, q_1} & 0 \\ 0 & B_{2, q_1}\end{bmatrix}
	\begin{bmatrix}G_{1,q_0} & 0 \\ 0 & G_{2,q_0}\end{bmatrix}
	\begin{bmatrix}	\id_m \\ \id_m\end{bmatrix}
     &= 0 \\
\implies 
&\forall d  \forall q_d \dots q_1 \nonumber& \\
&	\begin{bmatrix}	1 & -1\end{bmatrix}
	\begin{bmatrix}C_1 & 0 \\ 0 & C_2\end{bmatrix}
	\begin{bmatrix}B_{1, q_d} & 0 \\ 0 & B_{2, q_d}\end{bmatrix}
	\dots
	\begin{bmatrix}B_{1, q_1} & 0 \\ 0 & B_{2, q_1}\end{bmatrix}
	\begin{bmatrix}G_{1,q_0} & 0 \\ 0 & G_{2,q_0}\end{bmatrix}
	\begin{bmatrix}	\id_m \\ \id_m\end{bmatrix}
     &= 0
\end{align}

And therefore:

\begin{align}
&\interpret{C_1} \interpret{B_{1,q_d}} \dots \interpret{B_{1,q_1}} \interpret{G_{1,q_0}} =
\interpret{C_2} \interpret{B_{2,q_d}} \dots \interpret{B_{2,q_1}} \interpret{G_{2,q_0}} & \forall d \leq N \; \forall q_d \dots q_1 \\
\implies& 
\interpret{C_1} \interpret{B_{1,q_d}} \dots \interpret{B_{1,q_1}} \interpret{G_{1,q_0}} =
\interpret{C_2} \interpret{B_{2,q_d}} \dots \interpret{B_{2,q_1}} \interpret{G_{2,q_0}} & \forall d \; \forall q_d \dots q_1 
\end{align}

And therefore for any choice of value for $\alpha_j$ and $\delta_{k>1}$:

\begin{align}
&\interpret{\family{\bb{D}_1 }_{\alphadelta|\delta_1=d}} = \interpret{\family{ \bb{D}_2 }_{\alphadelta|\delta_1=d}} & \forall d \leq N \\
\implies&\interpret{\family{\bb{D}_1 }_{\alphadelta|\delta_1=d}} = \interpret{\family{ \bb{D}_2 }_{\alphadelta|\delta_1=d}} & \forall d
\end{align}
\end{proof}

\subsection{Proof of theorem \ref{thmfinite}}

\thmFinite*

\label{secfinite}
\begin{proof}

Throughout this proof we iterate on the set $\set{ \bbE_\kappa }$.
While at all times $\set{ \bbE_\kappa }$ is a finite verifying set for $\bbE$,
it is only at the end that $\set{ \bbE_\kappa }$ is a finite verifying set of simple equations.
We first show the existence of an ordered list of the !-boxes present in $E$,
compatible with the nesting order on both of the nesting diagrams of $E$.

\begin{definition}
We construct the ordered list $\delta_{k_1} \succ \delta_{k_2} \succ \dots \succ \delta_{k_n}$
by recursively picking a !-box that is nested in no other !-boxes,
then removing that !-box from the nesting diagram. Repeat on the new nesting diagram.
\end{definition}

\begin{definition} The algorithm \textbf{!-Removal}:
\begin{itemize}
	\item \textbf{If we are copying variable names:} \;
	Given a verifying set of equations $\set{\bbE_\kappa}_{\kappa \in K}$,
and a !-box $\delta_k$ nested in no other !-boxes present in the ${\bbE_\kappa}$:

We define !-Remove($\delta_k$) as the process described in theorem \ref{thmbbox2}.
It acts on the set $\set{\bbE_\kappa}_{\kappa \in K}$
by acting on each of the $\bbE_\kappa$ in turn,
finding the value $N_\kappa$, and
creating the new verification set:

\begin{align}\set{\bbE_\kappa}_{\kappa \in K'} :=
\bigcup_{\kappa \in K} \set{ \bbE_\kappa \at{\delta_k = 1}, \dots, \bbE_\kappa \at{\delta_k = N_\kappa} }
\end{align}

\item \textbf{If we are creating child instances:} \;
We do as above, but we pick the !-box $\delta_k$ 
such that all its child instances are nested in no other !-boxes present in the ${\bbE_\kappa}$,
and we act not only on each of the $\bbE_\kappa$ in turn but also on each of the child instances of $\delta_k$ in turn.
\end{itemize}

\end{definition}

\textbf{Claims:} \; 
\begin{itemize}
	\item !-Remove($\delta_k$) removes any dependency on $\delta_k$ in the verifying set $\set{\bbE_\kappa}_{\kappa \in K'}$
	\item $\set{\bbE_\kappa}_{\kappa \in K'}$ verifies $\set{\bbE_\kappa}_{\kappa \in K}$
	\item !-Remove($\delta_k$) does not alter the nesting ordering of any remaining !-boxes and phase variables
	in the verification pair
	\item The ordered list $\delta_{k_1} \succ \delta_{k_2} \succ \dots$ provides us with a sequence of !-boxes
	such that we can apply !-Remove($\delta_{k_{n+1}}$) to the output of !-Remove($\delta_{k_n}$).
	\item This ordered sequence of !-Removes results in a finite verifying set that has no dependence on any !-box.
\end{itemize}
Proof of claims: The first, second and fifth claims follow from Theorem~\ref{thmbbox2}.
The third and fourth of these claims are easy when we are copying variable names,
but when creating child instances of !-boxes below $\delta_k$
one should view instantiation as creating (distinctly named) copies of the nesting structure that
exists below $\delta_k$.

\vskip 2em

Removing phase variables is trickier than removing !-boxes,
because we know by theorem \ref{thmboth} that to do so in a manner compatible with !-box removal
we first need to know the largest number of !-box instances we are going to need.
Since we have already constructed the $\set{\bbE_\kappa}_{\kappa \in K}$ we can find the equation that resulted from
every !-box $\delta_k$ being instantiated at its largest amount, $N_k$, and use that equation.

\begin{definition} We use $\bbE'$ to denote the !-box free equation that is the result of instantiating
each !-box $\delta_k$ at its largest required amount $N_k$
\end{definition}

\textbf{Claim:} \; The equation $\bbE'$ 
contains the largest number of instances of $\alpha_j$ compared to any other equation $\bbE_\kappa$

Proof of claim: The number of instances of $\alpha_j$ is determined by the number of times the !-boxes
have been instantiated.

\vskip 2em

\begin{definition} The algorithm \textbf{$\alpha$-Removal}:
Given a verifying set $\set{\bbE_\kappa}_{\kappa \in K}$ we construct the set $A_j$ for variable $\alpha_j$ by considering the degree of $\alpha_j$ in $\bbE'$,
and then choosing enough valid values of $\alpha_j$ to reach the amount dictated by theorem \ref{thmparameter}.
We then form:

\begin{align}
\set{\bbE_\kappa}_{\kappa \in K'} :=
\bigcup_{\kappa \in K, a \in A_j} \set{ \bbE_\kappa \at{\alpha_j = a}}
\end{align}

\end{definition}

\textbf{Claim:} \; By theorem \ref{thmboth} $\set{\bbE_\kappa}_{\kappa \in K'}$ verifies $\set{\bbE_\kappa}_{\kappa \in K}$

And finally:

\textbf{Claim:} \, Starting with the verifying set $\set{\bbE}$: After applying !-Removes in the order dictacted by $\prec$ and then applying all possible $\alpha$-Removes
we construct a finite set $\set{\bbE_\kappa}_{\kappa \in K}$ of simple equations that verifies $\bbE$.

\end{proof}

\section{Generators and degrees}

We show here the generators of our three example languages (ZX, ZH and \ZW),
and then show the matrix degrees of these interpretations.

\label{secinterpretations}
\subsection{Generators}
\label{secgenerators}

\subsubsection{ZX}
\begin{itemize}
\item The Z spider has the following interpretation in $\Mat_{\bbC}$

\begin{align}
	\interpret{\family{\spider{gn}{\alpha}} \at {\alpha = a}} \quad = \quad 
\begin{bmatrix}
1 & 0 & \dots &  & 0 \\
0 & 0 & & & \vdots \\
\vdots & & \ddots & & \\
 & & & 0 & 0\\
0 & \dots & & 0 & e^{ia}
\end{bmatrix}
\end{align}

Making the substitution $Y := e^{i\alpha}$ (and therefore $Y^n = e^{ i n\alpha}$)
the Z spider has this interpretation in $\Mat_{\bbC[Y,Y\inv]}$:

\begin{align}
	\interpret{\spider{gn}{\alpha}} \quad = \quad 
\begin{bmatrix}
1 & 0 & \dots &  & 0 \\
0 & 0 & & & \vdots \\
\vdots & & \ddots & & \\
 & & & 0 & 0\\
0 & \dots & & 0 & Y
\end{bmatrix}
\end{align}

\item The Hadamard node has the following interpretation

\begin{align} \interpret{\node{H}{}} \quad = \quad
\begin{bmatrix}
1 & 1 \\
1 & -1
\end{bmatrix}
\end{align}

And admits no parameters.

\item The X spider interpretation is found by applying Hadamards nodes on all inputs and outputs of the corresponding Z spider.

\end{itemize}
\subsubsection{ZH}
\begin{itemize}
\item The H box in ZH has the following interpretation in $\Mat_\bbC$

\begin{align}
	\interpret{\family{\spider{ZH}{\alpha}}\at{\alpha = a}} \quad = \quad
\begin{bmatrix}
1 & 1 & \dots &  & 1 \\
1 & 1 & & & \vdots \\
\vdots & & \ddots & & \\
 & & & 1 & 1\\
1 & \dots & & 1 & a
\end{bmatrix}
\end{align}

We can simply equate $Y := \alpha$ and get the interpretation in $\Mat_{\bbC[Y,Y\inv]}$:

\begin{align}
	\interpret{\spider{ZH}{\alpha}} \quad = \quad
\begin{bmatrix}
1 & 1 & \dots &  & 1 \\
1 & 1 & & & \vdots \\
\vdots & & \ddots & & \\
 & & & 1 & 1\\
1 & \dots & & 1 & Y
\end{bmatrix}
\end{align}

\item The Z spider in ZH admits no parameters, and has the following interpretation

\begin{align}
	\interpret{\spider{white}{}} \quad = \quad
\begin{bmatrix}
1 & 0 & \dots &  & 0 \\
0 & 0 & & & \vdots \\
\vdots & & \ddots & & \\
 & & & 0 & 0\\
0 & \dots & & 0 & 1
\end{bmatrix}
\end{align}
\end{itemize}
\subsubsection{ZW}

We are explicitly using $ZW_\bb{C}$ here for compatibility with the other languages.
\begin{itemize}
	\item The Crossing $x$ has interpretation

\begin{align}
	\interpret{
\InputIfFileExists{crossing}{}{\input{./tikz/crossing.tikz}}
}
\begin{bmatrix}
1 & 0 & 0 & 0 \\
0 & 0 & 1 & 0 \\
0 & 1 & 0 & 0\\
0 & 0 & 0 & -1
\end{bmatrix}
\end{align}
	\item The W spider in bra-ket notation has interpretation

	\begin{align}
		\interpret{\wspider} \quad = \quad
	 \sum\limits_{k=1}^n \ket{\underbrace{0\ldots 0}_{k-1}1\underbrace{0\ldots 0}_{n-k}}
	 \end{align}

	\item The Z spider is parameterised by $\alpha \in \bb{C}$ and has interpretation 
	
\begin{align}
	\interpret{\family{\spider{white}{\alpha}}\at{\alpha = a}} \quad = \quad
\begin{bmatrix}
1 & 0 & \dots &  & 0 \\
0 & 0 & & & \vdots \\
\vdots & & \ddots & & \\
 & & & 0 & 0\\
0 & \dots & & 0 & a
\end{bmatrix}
\end{align}

	We can again simply equate $Y:= \alpha$ and get the interpretation in $\Mat_{\bbC[Y,Y\inv]}$:
	
\begin{align}
	\interpret{\spider{white}{\alpha}} \quad = \quad
\begin{bmatrix}
1 & 0 & \dots &  & 0 \\
0 & 0 & & & \vdots \\
\vdots & & \ddots & & \\
 & & & 0 & 0\\
0 & \dots & & 0 & Y
\end{bmatrix}
\end{align}

\end{itemize}

\subsection{degrees}

We will consider nodes parameterised by a single variable $\alpha$,
and express their degree with respect to $Y$ following the convention of appendix \ref{secgenerators}.

\subsubsection{ZX}

Using $Y^n:=e^{n i \alpha}$, the degrees in $Y$ of the generators (for $n \geq 0$) are:

\begin{align}
\maxdegree{+}{\alpha}{\spider{gn}{n \alpha}} &= n& \maxdegree{-}{\alpha}{\spider{gn}{n \alpha}} &= 0 \\
\maxdegree{+}{\alpha}{\spider{gn}{- n \alpha}} &= 0 &\maxdegree{-}{\alpha}{\spider{gn}{- n \alpha}} &= n \\
\maxdegree{+}{\alpha}{\spider{rn}{n \alpha}} &= n&\maxdegree{-}{\alpha}{\spider{rn}{n \alpha}} &= 0 \\
\maxdegree{+}{\alpha}{\spider{rn}{- n \alpha}} &= 0  &\maxdegree{-}{\alpha}{\spider{rn}{- n \alpha}} &= n \\
\maxdegree{+}{\alpha}{\node{H}{}} &= 0  &
\maxdegree{-}{\alpha}{\node{H}{}} &= 0
\end{align}

\subsubsection{ZH}

We equate $Y:= \alpha$, and for $P$ any Laurent polynomial:

\begin{align}
\maxdegree{+}{\alpha}{\spider{ZH}{P(\alpha)}} &= \maxdegree{+}{Y} P &
\maxdegree{-}{\alpha}{\spider{ZH}{P(\alpha)}} &= \maxdegree{-}{Y} P \\
\maxdegree{+}{\alpha}{\spider{white}{}} &= 0 & \maxdegree{-}{\alpha}{\spider{white}{}} &= 0
\end{align}

\subsubsection{ZW}

We equate $Y:= \alpha$, and for $P$ any Laurent polynomial:

\begin{align}
\maxdegree{+}{\alpha}{\spider{black}{}} &= 0  & \maxdegree{-}{\alpha}{\spider{black}{}} &= 0 \\
\maxdegree{+}{\alpha}{
\InputIfFileExists{crossing}{}{\input{./tikz/crossing.tikz}}
} &= 0 & \maxdegree{-}{\alpha}{
\InputIfFileExists{crossing}{}{\input{./tikz/crossing.tikz}}
} &= 0  \\
\maxdegree{+}{\alpha}{\spider{white}{P(\alpha)}} &= \maxdegree{+}{Y} P & 
\maxdegree{-}{\alpha}{\spider{white}{P(\alpha)}} &= \maxdegree{-}{Y} P \\
\end{align}

\section{Separability} \label{secSeparability}

We describe a non-separated pair of !-boxes as \textbf{separable} if we can perform the following operation:
\begin{align}

\InputIfFileExists{sep-l}{}{\input{./tikz/sep-l.tikz}}
 \quad = 
\InputIfFileExists{sep-r}{}{\input{./tikz/sep-r.tikz}}

\end{align}

We define pairs of nodes as separable if we can always separate !-boxes that are joined by edges between these pairs of nodes.
Note that we only need to consider nodes that have arbitrary arity,
since only they can be connected to !-boxes.

\begin{itemize}
\item The following pairs of nodes are separable, by language:
\begin{align}
\text{ZX:} \quad &
(\spider{gn}{}, \spider{rn}{}) \quad 
(\spider{gn}{}, \spider{gn}{}) \quad 
(\spider{rn}{}, \spider{rn}{}) \\
\text{ZH:} \quad &
(\spider{white}{}, \spider{white}{}) \quad 
(\spider{white}{}, \spider{ZHGrey}{}) \quad 
(\spider{ZHGrey}{}, \spider{ZH}{}) \\
\text{ZW:} \quad &
(\spider{white}{}, \spider{white}{}) \quad 
(\spider{white}{}, \spider{black}{})
\end{align}

The proofs of these follow immediately from the spider and bialgebra laws in the respective languages.

\item The following pairs of nodes are assumed to not be separable, by language:
\begin{align}
\text{ZX:} \quad & \text{Always separable} \\
\text{ZH:} \quad & 
(\spider{ZH}{}, \spider{ZH}{}) \quad 
(\spider{ZH}{}, \spider{white}{}) \\
\text{ZW:} \quad &
(\spider{black}{}, \spider{black}{})
\end{align}
\end{itemize}

Note that it is enough to specify the phase-free versions of these interactions,
because phases can always be moved away from the critical nodes.

\end{document}